%% file: main.tex
\documentclass[a4paper,11pt]{article}

\usepackage{graphicx}

\usepackage[margin=1in]{geometry}
\usepackage{amsmath,amsthm,amsfonts,amssymb,bm,bbm}
\usepackage{mathtools}
\usepackage{hyperref}
\usepackage[svgnames]{xcolor}
\hypersetup{colorlinks={true},urlcolor={blue},linkcolor={DarkBlue},citecolor={DarkGreen}}

\usepackage{microtype}
\usepackage[capitalise,nameinlink]{cleveref}

\usepackage{authblk}

\usepackage{doi}

\newtheorem{theorem}{Theorem}
\newtheorem{lemma}{Lemma}
\newtheorem{fact}{Fact}

\theoremstyle{definition}
\newtheorem{definition}{Definition}[section]
\newtheorem{remark}{Remark}

\newcommand{\bb}{\bm{b}}

\newcommand{\bx}{\bm{x}}
\newcommand{\by}{\bm{y}}

\newcommand{\bR}{\mathbb{R}}

\renewcommand{\epsilon}{\varepsilon}

\title{The Complexity of Symmetric Bimatrix Games \\ with Common Payoffs}
\author{Abheek Ghosh and Alexandros Hollender}
\affil{University of Oxford, UK}
\date{}

\begin{document}

\maketitle

\begin{abstract}
We study symmetric bimatrix games that also have the common-payoff property, i.e., the two players receive the same payoff at any outcome of the game. Due to the symmetry property, these games are guaranteed to have symmetric Nash equilibria, where the two players play the same (mixed) strategy. While the problem of computing such symmetric equilibria in general symmetric bimatrix games is known to be intractable, namely PPAD-complete, this result does not extend to our setting. Indeed, due to the common-payoff property, the problem lies in the lower class CLS, ruling out PPAD-hardness. In this paper, we show that the problem remains intractable, namely it is CLS-complete. On the way to proving this result, as our main technical contribution, we show that computing a Karush-Kuhn-Tucker (KKT) point of a quadratic program remains CLS-hard, even when the feasible domain is a simplex.
\end{abstract}

\input{intro}
\input{prelim}

\input{games-QPs}
\input{reduction}

\section{Open Problems}

Although the computation of equilibria in general bimatrix games is known to be PPAD-complete, there are some interesting special cases for which the complexity remains open. For example, it is known that \emph{win-lose} games, where all payoffs are 0 or 1, are hard~\cite{ChenDT09-Nash}. Separately, it is also known that \emph{sparse} bimatrix games, where the number of non-zero entries in each row and column is constant, are also hard~\cite{ChenDT09-Nash}. However, the complexity of bimatrix games that are both sparse and win-lose remains open~\cite{LiuS18-sparse-win-lose}. In a certain sense, these are the simplest bimatrix games, and understanding their complexity is an intriguing question.

\subsection*{Acknowledgments}
The authors wish to thank the reviewers for suggestions that helped improve the presentation of the paper. A.G.\ was supported by EPSRC Grant EP/X040461/1 ``Optimisation for Game Theory and Machine Learning''. A.H.\ thanks Koosha Samieefar for pointing out the connection with KKT points of quadratic programming shown in \cite{McLennanT10-complexity-imitation}.

\bibliographystyle{splncs04}
\bibliography{ref}

\end{document}

%% file: intro.tex
\section{Introduction}

The Nash equilibrium notion is an extremely natural solution concept for games. First of all, it has a very simple definition: every player should be happy with what they are playing at equilibrium, given what the other players are playing. In particular, this is something that each player can easily check for themselves. A second, arguably equally important, property is that every game is guaranteed to have a Nash equilibrium~\cite{Nash50,Nash51}. As a result, Nash equilibria have been used as the de facto standard solution concept in a vast array of economic settings.

One crucial question, which was somewhat overlooked for a long time, is whether Nash equilibria are easy to find. The advent of computer science made this question precise and formal: \emph{is there a polynomial-time algorithm for computing a Nash equilibrium of a game}? Given that, despite significant efforts, no such efficient algorithm had been found, it was natural to suspect that the problem is in fact intractable.
The theory of NP-completeness, although very successful, can unfortunately not be used for this problem. Interestingly, it is precisely the two core properties mentioned above -- efficient verifiability and guaranteed existence -- which make it essentially impossible for the problem to be NP-hard. Indeed, the problem lies in the class TFNP of total NP search problems, and no TFNP problem can be NP-hard, unless NP $=$ co-NP \cite{MegiddoP91-TFNP}. The complexity class PPAD, the correct class -- as it turned out -- to capture the complexity of computing Nash equilibria, was only defined in the early 90s~\cite{Papadimitriou94-TFNP-subclasses}, and the celebrated results of Daskalakis, Goldberg, and Papadimitriou~\cite{DaskalakisGP09-Nash}, and Chen, Deng, and Teng~\cite{ChenDT09-Nash} finally proved that the problem is PPAD-complete.

The PPAD-completeness of computing Nash equilibria provides strong evidence that the problem cannot be solved in polynomial time. Indeed, many different problems with no known efficient algorithms are known lie in this class. More recently, the hardness of PPAD has also been shown assuming various cryptographic assumptions~\cite{BitanskyPR15-Nash-crypto,ChoudhuriHKPRR19-Fiat-Shamir,JawaleKKZ21-PPAD-LWE}. This intractability result brings into question whether the Nash equilibrium notion is in fact a reasonable solution concept for these games. Indeed, the PPAD-completeness result~\cite{ChenDT09-Nash} even applies to bimatrix games, namely simple normal-form games between just two players.

In the past two decades, an important research direction has consisted in investigating whether Nash equilibria can at least be computed efficiently for special restricted classes of games. A very natural such restriction of bimatrix games are \emph{symmetric} bimatrix games. These are games in which nothing changes if we swap the labels of the two players. Formally, the payoff matrix of the second player is just the transpose of the payoff matrix of the first player. This is a very natural condition, which is in particular satisfied by many standard games such as the prisoner's dilemma, the stag hunt, or the game of chicken. Nash proved that symmetric games always admit \emph{symmetric} equilibria~\cite{Nash51}, where both players play the same mixed strategy. As a result, it is natural to only consider equilibria that have this desirable property in this setting. Unfortunately, this problem remains PPAD-complete~\cite[Theorem 2.4]{Papadimitriou07-AGT-complexity-equilibria}. Other notable negative results for bimatrix games include the PPAD-hardness of sparse games~\cite{ChenDT09-Nash}, win-lose games~\cite{ChenDT09-Nash}, and rank-3 games~\cite{Mehta18-constant-rank}. Some positive results have been obtained for restrictions such as, of course, zero-sum games~\cite{Neumann1928-zero-sum}, but also rank-1 bimatrix games~\cite{AdsulGMS11-rank-one}. Interestingly, all these results have shown either PPAD-completeness or polynomial-time solvability. This raises the question of whether there are natural classes of bimatrix games for which the problem is neither PPAD-complete, nor in P.

\paragraph{\bf Our contribution.}
In this work, we identify for the first time a natural class of bimatrix games for which the equilibrium problem has \emph{intermediate} complexity. Namely, we show that the problem of computing a Nash equilibrium in such games is CLS-complete, meaning that it is unlikely to be PPAD-complete or polynomial-time solvable.

The class of games we consider are \emph{symmetric bimatrix games with common payoffs}, meaning that the two players have the same payoff for any given outcome of the game. More formally, these are symmetric games $(A,A^\intercal)$ that also satisfy $A = A^\intercal$, and we consider symmetric equilibria in such games.\footnote{Since these games have common payoffs, a pure Nash equilibrium is guaranteed to exist and can be easily located by simply picking the maximum entry in the payoff matrix $A$. However, note that such pure equilibria will usually \emph{not} be symmetric.} These games, studied by Emmons et al.~\cite{Emmons22-symmetric-common-payoff}, can be used to model settings in multiagent reinforcement learning and cooperative AI, although mostly with a larger number of players. Since our hardness result already holds for only two players, it also extends to these settings.\footnote{Concurrent work by Tewolde et al.~\cite{TewoldeZOSC25-game-symmetries} establishes a weaker CLS-hardness result (for five players instead of two), and a stronger CLS-membership result for many players and more general symmetries.}

By using a well-known connection~\cite{McLennanT10-complexity-imitation,MurhekarM20-imitation}, our result immediately also implies the CLS-completeness of computing (not necessarily symmetric) Nash equilibria in imitation games $(A, I)$ with symmetric matrix $A$.\footnote{Here $I$ denotes the identity matrix.}

Our result is obtained by leveraging a known connection between equilibria of these games and Karush-Kuhn-Tucker (KKT) points of quadratic programs of a particular form. This connection was already observed and used by McLennan and Tourky~\cite{McLennanT10-complexity-imitation} in their work on NP-hardness of various decision problems related to Nash equilibria. In our context, it immediately yields CLS-membership for our equilibrium computation problem. Our main technical contribution is to also establish CLS-hardness of the problem. This is achieved by showing that the problem of computing a KKT point of a quadratic program \emph{with simplex constraints} is CLS-complete. Recent work of Fearnley et al.~\cite{fearnley2024complexity} had established CLS-completeness of the problem \emph{with box constraints}, and had left open the case of simplex constraints. We present a direct reduction from box constraints to simplex constraints, thus resolving this case as well.

\paragraph{\bf Continuous Local Search (CLS).}
The class CLS, introduced by Daskalakis and Papadimitriou~\cite{DaskalakisP11-CLS}, lies both in PPAD and in PLS. The class PLS captures the complexity of various problems in discrete local optimization, such as the local max-cut problem~\cite{JPY1988-PLS,Krentel1989-TSP,Schaeffer1991local,FabrikantPT2004pure}. Fearnley et al.~\cite{FearnleyGHS22-gradient} proved that CLS $=$ PPAD $\cap$ PLS, and that CLS exactly captures the complexity of finding gradient descent fixed points, i.e., KKT points of general smooth functions given by circuits. Babichenko and Rubinstein~\cite{BabichenkoR21-congestion} further showed CLS-completeness of computing mixed Nash equilibria of congestion games with many players, as well as of finding KKT points of polynomials of degree five. Recently, Fearnley et al.~\cite{fearnley2024complexity} showed that this remains the case even for polynomials of degree two, i.e., quadratic programming. Their result only holds for box constraints.

A CLS-completeness result indicates that the problem is unlikely to lie in P, since the cryptographic hardness results mentioned above for PPAD also apply to CLS. It is also believed that CLS $\neq$ PPAD. Indeed, it is known that no black-box reduction from PPAD to CLS exists~\cite{Morioka01-Mthesis-PLS}, and, so far, all collapses of classes in the context of TFNP have used black-box reductions.

\paragraph{\bf Further related work.}
Kontogiannis and Spirakis~\cite{KontogiannisS11-symmetric-bimatrix} use KKT points of an appropriately defined quadratic formulation of symmetric bimatrix games to obtain almost $1/3$-approximate Nash equilibria. It is known that computing Nash equilibria in imitation games is PPAD-complete~\cite{CodenottiS05-bimatrix-zero-one-imitation,ChenDT09-Nash}. Furthermore, Murhekar and Mehta~\cite{MurhekarM20-imitation} show that although the problem does not admit a FPTAS (Fully Polynomial Approximation Scheme), there is a PTAS. Various decision problems about the existence of Nash equilibria with certain properties in bimatrix games have been proved NP-complete, even for symmetric or imitation games~\cite{GilboaZ1989-Nash-decision,ConitzerS08-Nash,McLennanT10-complexity-imitation}.

\paragraph{\bf Outline.}
In \cref{sec:prelims}, we present some basic definitions and standard results about bimatrix games, and we introduce the KKT problem for quadratic programs. In \cref{sec:games-QPs}, we explain the connection between the two, and state our results for games. Finally, in \cref{sec:reduction}, we present our main technical result, namely the CLS-hardness of KKT points in quadratic programs with simplex constraints.

%% file: prelim.tex
\section{Preliminaries}\label{sec:prelims}
Let $[n] = \{1, 2, \ldots, n\}$. Let $\Delta^{n-1} \subset \bR^n$ denote the $(n-1)$-dimensional simplex in $\bR^n$, i.e., $\bx \in \Delta^{n-1}$ if $x_i \ge 0$ for all $i \in [n]$ and $\sum_i x_i = 1$.

\subsection{Game Theory: Bimatrix Games}

\begin{definition}
Let $n,m \in \mathbb{N}$. A bimatrix game $(A,B)$ is given by two matrices $A = (a_{ij}) \in [0,1]^{n \times m}$ and $B = (b_{ij}) \in [0,1]^{n \times m}$, with the interpretation that player 1 (the row player) and player 2 (the column player) have $n$ and $m$ pure strategies respectively. If the row player plays its $i$th strategy, and the column player plays its $j$th strategy, then the payoff to the row player is $a_{ij}$, and the payoff to the column player is $b_{ij}$.

A bimatrix game is
\begin{itemize}
\item \emph{symmetric}, if $B = A^\intercal $ (in particular, $n=m$).
\item \emph{common-payoff}, if $B = A$.
\item an \emph{imitation} game, if $n=m$ and $B = I_n$ (the $n \times n$ identity matrix).
\end{itemize}
\end{definition}

\begin{definition}
Let $\varepsilon \geq 0$. A \emph{mixed strategy profile} $(\bx,\by) \in \Delta^{n-1} \times \Delta^{m-1}$ is an $\varepsilon$-Nash equilibrium if
\begin{align*}
\bx^\intercal A\by \geq \hat{\bx}^\intercal A\by - \varepsilon, \quad &\text{ for all } \hat{\bx} \in \Delta^{n-1},\\
\bx^\intercal B\by \geq \bx^\intercal B\hat{\by} - \varepsilon, \quad &\text{ for all } \hat{\by} \in \Delta^{m-1}.
\end{align*}
The equilibrium is \emph{symmetric}, if $n=m$ and $\bx=\by$.
\end{definition}

\begin{fact}[Nash~\cite{Nash51}]
Every bimatrix game has a Nash equilibrium. Every symmetric bimatrix game has a symmetric Nash equilibrium.
\end{fact}

The following alternative, stronger way of defining approximation of Nash equilibria will be useful in the paper.

\begin{definition}
Let $\varepsilon \geq 0$. A mixed strategy profile $(\bx,\by) \in \Delta^{n-1} \times \Delta^{m-1}$ is an $\varepsilon$-\emph{well-supported} Nash equilibrium if for all $i \in [n]$
\[x_i > 0 \implies (A\by)_i \geq \max_j (A\by)_j - \varepsilon\]
and for all $i \in [m]$
\[y_i > 0 \implies (\bx^\intercal  B)_i \geq \max_j (\bx^\intercal  B)_j - \varepsilon.\]
\end{definition}

It is easy to see that any $\varepsilon$-\emph{well-supported} Nash equilibrium is also an $\varepsilon$-Nash equilibrium, but the converse is not necessarily true. Nevertheless, the notions are known to be computationally equivalent in the following sense.

\begin{lemma}[{Chen et al.~\cite[Lemma~3.2]{ChenDT09-Nash}}]
\label{lem:well-supported}
In a bimatrix game $(A,B)$ with $A,B \in [0,1]^{n \times m}$, given any $\varepsilon^2/8$-Nash equilibrium, we can find an $\varepsilon$-well-supported Nash equilibrium in polynomial time.
\end{lemma}

This continues to hold in the case of \emph{symmetric} equilibria in symmetric bimatrix games.\footnote{It suffices to inspect the proof of \cite[Lemma~3.2]{ChenDT09-Nash} to see that for symmetric bimatrix games, the symmetry of a strategy profile is retained in the polynomial-time transformation.}

\subsection{Optimization: The KKT Problem for Quadratic Programs}\label{sec:prelim:optim}
\begin{definition}[Karush--Kuhn--Tucker (KKT) Conditions]
Consider the following constrained optimization problem in $\bR^n$
\begin{align}
    \min_{\bx \in \bR^n} \qquad\qquad f(\bx) \label{eq:optim:1} \\
    \text{ subject to} \qquad\qquad g_j(\bx) &\le 0, &\text{for $j \in [m]$.} \label{eq:optim:2} 
\end{align}
We assume that the function $f$ is continuously differentiable and the functions $g_j$ for $j \in [m]$ are continuously differentiable and convex. The Karush--Kuhn--Tucker (KKT) conditions are:
\begin{align}
    \nabla f(\bx) + \sum_j u_j \nabla g_j(\bx) &= 0, \label{eq:optim:kkt:1} \\
    u_j g_j(\bx) &= 0, &\text{for $j \in [m]$,} \label{eq:optim:kkt:2} \\
    u_j &\ge 0, &\text{for $j \in [m]$,} \label{eq:optim:kkt:3}
\end{align}
and $\bx$ must also satisfy the original constraints in \eqref{eq:optim:2}.
The $u_j$ variables are called the dual variables.
A point $\bx$ that satisfies the KKT conditions above is called a KKT point.
\end{definition}
A KKT point is stable with respect to gradient descent and vice-versa. In gradient descent, we move towards the direction of decreasing gradient of $f(\bx)$, i.e., towards a direction that has positive inner-product with $-\nabla f(\bx)$, unless by doing so we violate one of the constraints. The $u_j g_j(\bx) = 0$ condition in \eqref{eq:optim:kkt:2} ensures that $u_j > 0$ only if $g_j(\bx) = 0$, i.e., $u_j$ can take positive values only if the corresponding constraint is a tight one. Then, $\sum_j u_j \nabla g_j(\bx)$ corresponds to the convex cone of the gradients of the tight constraints, and the condition \eqref{eq:optim:kkt:1} ensures that $-\nabla f(\bx)$ is in this cone; therefore, by moving towards a direction that decreases $f(\bx)$, we violate a constraint.

\begin{definition}[$\varepsilon$-approximate KKT Conditions]
The $\varepsilon$-approximate KKT (or simply $\varepsilon$-KKT) conditions relax the KKT conditions in \eqref{eq:optim:kkt:1}--\eqref{eq:optim:kkt:3} (for the optimization problem in \eqref{eq:optim:1}--\eqref{eq:optim:2}):
\begin{align}
    \nabla f(\bx) + \sum_j u_j \nabla g_j(\bx) &\in [-\varepsilon, \varepsilon], \label{eq:optim:kkt:approx:1} \\
    u_j g_j(\bx) &= 0, &\text{for $j \in [m]$,} \label{eq:optim:kkt:apx:2} \\
    u_j &\ge 0,  &\text{for $j \in [m]$.} \label{eq:optim:kkt:apx:3}
\end{align}
\end{definition}

In this paper, we primarily focus on quadratic programs (QPs) with linear constraints. In such a QP, the objective function $f(\bx)$ in \eqref{eq:optim:1} is a quadratic function and the functions $g_j(\bx)$ in the constraints in \eqref{eq:optim:2} are linear functions. In particular, we look at box constraints and simplex constraints.

\begin{definition}[QP with Box Constraints]
A quadratic program with box constraints is given as
\begin{align}
    \min_{\bx}& \qquad\qquad \frac{1}{2}\bx^\intercal A \bx + \bb^\intercal \bx \label{eq:box_qp:1} \\
    \text{ subject to}& \qquad\qquad 0 \le x_i \le 1, \qquad\qquad\text{ for $i \in [n]$,} \label{eq:box_qp:2}
\end{align}
where  $A = (a_{ij}) \in \bR^{n \times n}$ is symmetric without loss of generality, $a_{ij} = a_{ji}$, and $\bb = (b_i) \in \bR^n$.
\end{definition}
The $\varepsilon$-KKT conditions for a QP with box constraints can be simplified to the following form:
\begin{definition}[$\varepsilon$-KKT Conditions for a QP with Box Constraints]
A feasible point $\bx$ satisfies the $\varepsilon$-KKT Conditions if 
\begin{align}
    \text{for } i \in [n] : \begin{cases}
        \sum_j a_{ij} x_j + b_i \ge -\varepsilon,  &\text{ if } x_i = 0, \\
        \sum_j a_{ij} x_j + b_i \in [-\varepsilon, \varepsilon], &\text{ if } 0 < x_i < 1, \\
        \sum_j a_{ij} x_j + b_i \le \varepsilon, &\text{ if } x_i = 1.
    \end{cases} \label{eq:box_kkt}
\end{align}
\end{definition}
It is known that computing an $\varepsilon$-KKT point of a QP with box constraints is CLS-complete~\cite{fearnley2024complexity}.

\begin{definition}[QP with a Simplex Constraint]
A quadratic program with simplex constraints is given as
\begin{align}
    \min_{\bx}& \qquad\qquad \frac{1}{2} \bx^\intercal A \bx + \bb^\intercal \bx \label{eq:simplex_qp:1} \\
    \text{ subject to}& \qquad\qquad\sum_i x_i = 1,  \label{eq:simplex_qp:2:normalized} \\
    & \qquad\qquad \ x_i \ge 0,  \qquad\qquad\text{ for $i \in [n]$.} \label{eq:simplex_qp:3}
\end{align}
where  $A = (a_{ij}) \in \bR^{n \times n}$ is symmetric without loss of generality, $a_{ij} = a_{ji}$, and $\bb = (b_i) \in \bR^n$.
\end{definition}
Note that a QP with a scaled version of the simplex constraint of the form
\begin{equation}\label{eq:simplex_qp:2}
    \sum_i x_i = s
\end{equation}
for any fixed $s > 0$ is equivalent to the canonical form: we can change variables from $x_i$ to $x_i' = x_i/s$ and the coefficients from $a_{ij}$ to $a_{ij}' = a_{ij} s^2$ and from $b_i$ to $b_i' = b_i s$ to recover the canonical form. We shall use this scaled version in some of our proofs for convenience.

\begin{definition}[$\varepsilon$-KKT Conditions for a QP with a Simplex Constraint]
A feasible point $\bx$ satisfies the $\varepsilon$-KKT Conditions if there exists $u \in \bR$ such that
\begin{align}
    \text{for } i \in [n] : \begin{cases}
        \sum_j a_{ij} x_j + b_i \ge u - \varepsilon &\text{ if } x_i = 0, \\
        \sum_j a_{ij} x_j + b_i \in [u - \varepsilon, u + \varepsilon] &\text{ if } x_i > 0.
    \end{cases} \label{eq:simplex_kkt}
\end{align}
$u \in \bR$ is the dual variable corresponding to the $\sum_i x_i = 1$ constraint.
\end{definition}
One of the main contributions of our paper is to show that computing an $\varepsilon$-KKT point for simplex-constrained QPs is CLS-complete through a reduction from the same computational problem for box-constrained QPs.

%% file: games-QPs.tex
\section{The Complexity of Symmetric Common-Payoff Games}\label{sec:games-QPs}

The main result of this section is the following.

\begin{theorem}\label{thm:symmetric}
The problem of computing a symmetric $\varepsilon$-Nash equilibrium in a symmetric bimatrix game with common payoffs is CLS-complete.
\end{theorem}

\begin{remark}
Here, and everywhere else in this paper, when we talk about $\varepsilon$-Nash equilibria (or $\varepsilon$-well-supported Nash equilibria, or $\varepsilon$-KKT points), we always consider the setting where $\varepsilon$ is provided as part of the input, in standard binary representation. In particular, this means that all our hardness results apply to the setting where $\varepsilon$ is \emph{inverse-exponential} in the size of the input. Naturally, our results also apply to exact solutions,\footnote{For the games and KKT problems we consider, it can be shown using standard techniques (see, e.g.,~\cite{EtessamiY10-FIXP}) that they admit exact rational solutions with polynomial-length representations, and that the exact versions are equivalent to the inverse-exponential $\varepsilon$ versions.} namely to the case where $\varepsilon = 0$. Finally, when $\varepsilon$ is inverse-polynomial, an $\varepsilon$-equilibrium can be computed in polynomial time by using gradient descent on the optimization problem given in \cref{lem:game-KKT}. In other words, the problem admits a FPTAS and thus our hardness result for inverse-exponential $\varepsilon$ is optimal.
\end{remark}

\begin{proof}
By \cref{lem:well-supported}, it suffices to establish this result for approximate well-supported equilibria of the game. By \cref{lem:game-KKT}, stated and proved below, the problem is equivalent to computing an approximate KKT point of a quadratic program with a simplex constraint. Finally, by \cref{thm:QP}, our main technical result which is stated and proved in the next section, this problem is CLS-complete. Note that we can ensure that all the payoffs lie in $[0,1]$ by a standard rescaling and offset argument.
\end{proof}

The following lemma was proved by McLennan and Tourky~\cite{McLennanT10-complexity-imitation} for exact equilibria. Here we show that the result continues to hold if we introduce approximation.

\begin{lemma}\label{lem:game-KKT}
Let $(A,B)$ be a symmetric bimatrix game with common payoffs, i.e., $B = A^\intercal = A$. Then, a strategy profile $(\bx,\bx)$ is an $\varepsilon$-well-supported Nash equilibrium, if and only if $\bx$ is an $\varepsilon$-KKT point of
\begin{align*}
    \min_{\bx} \qquad\qquad - \bx^\intercal A &\bx \\
    \text{ subject to} \qquad\qquad\sum_i x_i &= 1,  \\
     \qquad\qquad\quad x_i &\ge 0, \qquad\qquad\text{ for $i \in [n]$.}
\end{align*}
\end{lemma}

\begin{proof}
Using the fact that the matrix $A = (a_{ij})$ is symmetric, we can rewrite
\[
- \bx^\intercal A \bx = - 2 \sum_{i < j} a_{ij} x_i x_j - \sum_{i} a_{ii} x_i^2.
\]
and thus
\[
\frac{\partial}{\partial x_i} \left( - \bx^\intercal A \bx \right) = - 2 \sum_{j} a_{ij} x_j = -2 (A\bx)_i.
\]
As a result, $\bx$ satisfies the $\varepsilon$-KKT conditions if any only if there exists $u \in \bR$ such that
\begin{align*}
    \text{for } i \in [n] : \begin{cases}
        -2 (A\bx)_i \ge u - \varepsilon &\text{ if } x_i = 0, \\
        -2 (A\bx)_i \in [u - \varepsilon, u + \varepsilon] &\text{ if } x_i > 0,
    \end{cases}
\end{align*}
or equivalently, if and only if there exists $v \in \bR$ such that
\begin{align*}
    \text{for } i \in [n] : \begin{cases}
        (A\bx)_i \le v + \varepsilon/2 &\text{ if } x_i = 0, \\
        (A\bx)_i \in [v - \varepsilon/2, v + \varepsilon/2] &\text{ if } x_i > 0.
    \end{cases}
\end{align*}
Now, it is not hard to see that such $v \in \bR$ exists if and only if
\[
(A\bx)_i \geq \max_j (A\bx)_j - \varepsilon \qquad \text{for all $i$ with $x_i > 0$.}
\]
But this is exactly the definition of $(\bx,\bx)$ being an $\varepsilon$-well-supported Nash equilibrium of the game $(A,A)$.
\end{proof}

\subsection{Consequences for Imitation Games}

The CLS-completeness for symmetric common-payoff games immediately yields the following result.

\begin{theorem}\label{thm:imitation}
The problem of computing an $\varepsilon$-Nash equilibrium in an imitation game $(A,I)$ with $A^\intercal=A$ is CLS-complete.
\end{theorem}

\begin{proof}
As before, by \cref{lem:well-supported}, it suffices to establish this result for approximate well-supported equilibria of the game. By \cref{lem:imitation}, stated and proved below, this problem is equivalent to that of finding a symmetric approximate well-supported Nash equilibrium in the symmetric game $(A,A^\intercal)$ with $A^\intercal=A$, i.e., with common payoffs. The result then follows by \cref{thm:symmetric}.
\end{proof}

The following lemma is well-known for exact equilibria, see e.g.~\cite{McLennanT10-complexity-imitation}. We present a version that also holds for approximate equilibria.

\begin{lemma}\label{lem:imitation}
Let $A$ be a square $n \times n$ matrix, and let $\varepsilon < 1/n$. Then the following are equivalent for all $\by$:
\begin{enumerate}
\item[(1)] $(\by,\by)$ is a symmetric $\varepsilon$-well-supported Nash equilibrium of the symmetric game $(A,A^\intercal)$,
\item[(2)] there exists $\bx$ such that $(\bx,\by)$ is an $\varepsilon$-well-supported Nash equilibrium of the imitation game $(A,I)$.
\end{enumerate}
\end{lemma}

\begin{proof}
First, assume that $(\by,\by)$ is a symmetric $\varepsilon$-well-supported Nash equilibrium of the symmetric game $(A,A^\intercal)$. Then, we have that for all $i \in [n]$
\[y_i > 0 \implies (A\by)_i \geq \max_j (A\by)_j - \varepsilon.\]
Now, define strategy $\bx$ to be the uniform distribution on $S := \{i \in [n]: y_i > 0\}$. Then, we have that for all $i \in [n]$
\[x_i > 0 \implies y_i > 0 \implies (A\by)_i \geq \max_j (A\by)_j - \varepsilon\]
and
\[y_i > 0 \implies x_i = 1/|S| = \max_j x_j \implies (I\bx)_i = \max_j (I\bx)_j\]
i.e., $(\bx,\by)$ is an $\varepsilon$-well-supported Nash equilibrium of the imitation game $(A,I)$.

Conversely, assume that there exists $\bx$ such that $(\bx,\by)$ is an $\varepsilon$-well-supported Nash equilibrium of the imitation game $(A,I)$. Then, we have that for all $i \in [n]$
\[y_i > 0 \implies (I\bx)_i \geq \max_j (I\bx)_j - \varepsilon \implies x_i \geq \max_j x_j - \varepsilon \geq 1/n - \varepsilon > 0\]
where we used the fact that $\max_j x_j \geq 1/n$, since $\bx$ is a probability distribution. Thus, we obtain
\[y_i > 0 \implies x_i > 0 \implies (A\by)_i \geq \max_j (A\by)_j - \varepsilon\]
which means that $(\by,\by)$ is a symmetric $\varepsilon$-well-supported Nash equilibrium of the symmetric game $(A,A^\intercal)$.
\end{proof}

%% file: reduction.tex
\section{CLS-hardness of the KKT Problem for QPs on the Simplex}
\label{sec:reduction}

We show that computing approximate KKT points for simplex-constrained QPs is CLS-complete through a reduction from the same computational problem for box-constrained QPs.

\begin{theorem}\label{thm:QP}
The problem of computing an $\varepsilon$-KKT point of a quadratic program with a simplex constraint and no first-order terms in the objective function
\begin{align}
    \min_{\bx}& \qquad\qquad \frac{1}{2}\bx^\intercal A\bx \\
    \text{ subject to}& \qquad\qquad\sum_i x_i = 1, \\
    & \qquad\qquad\ x_i \ge 0, \qquad\qquad\text{ for $i \in [n]$,}
\end{align}
where $A = (a_{ij}) \in \bR^{n \times n}$ is symmetric without loss of generality, $a_{ij} = a_{ji}$, is CLS-complete.
\end{theorem}
\begin{proof}
Notice that the simplex-constrained QP in the theorem statement does not have any first-order terms of the form $b_i x_i$. This is without loss of generality because $b_i x_i$ can be replaced by $b_i x_i (\sum_j x_j)$ to make all terms second-order (as $\sum_j x_j = 1$). So, the two formulations are equivalent. Similarly, the version where the matrix $A$ is not symmetric is equivalent to the version where it is symmetric, since we can just replace $A$ by $(A+A^\intercal)/2$, without changing the objective function.

Further, as discussed in \cref{sec:prelim:optim}, a scaled simplex constraint of the form $\sum_j x_j = s$ for some $s > 0$ instead of $\sum_j x_j = 1$ also keeps the problem effectively unchanged (by suitable change of variables). So, for convenience, we will use QPs that include first-order terms and have scaled simplex constraints in our reduction.

We start with an arbitrary instance of a QP with box constraints; we assume that $A = (a_{ij}) \in \bR^{n \times n}$ and $\bb = (b_i) \in \bR^n$ are the coefficients of this box-constrained QP, the same notation as \eqref{eq:box_qp:1}. We want to compute an $\varepsilon$-approximate KKT point of this QP. 
We reduce this to the problem of computing a $\delta$-approximate KKT point of a polynomially larger simplex-constrained QP and where $\delta = \Theta(\varepsilon)$.

Let $M = \max(1, \max_i (|b_i| + \sum_{j} |a_{ij}|))$ and $\delta = \varepsilon/(4 + 4nM)$. We construct the following QP with $2n+1$ variables $\bx = (x_1, \ldots, x_n)$, $\by = (y_1, \ldots, y_n)$, and $z$:
\begin{align}
    \min_{\bx, \by, z} \qquad\sum_{i < j} a_{ij} x_i x_j + \sum_{i} \frac{a_{ii}}{2} x_i^2 &+ \sum_i b_i x_i + \frac{M}{2\delta}\sum_i (x_i + y_i - 1)^2 \label{eq:qp:1} \\
    \text{ subject to} \qquad\qquad\sum_i x_i + \sum_i y_i + z &= 2n,  \label{eq:qp:2} \\
     x_i, y_i, z &\ge 0, \qquad\text{ for $i \in [n]$.} \label{eq:qp:3}
\end{align}
The coefficients $a_{ij}$ and $b_i$ are the same ones as the original box-constrained QP. The intuition for this constructed simplex-constrained QP is as follows: The $x_i$ variables will approximately correspond to the original $n$ variables of the box-constrained QP. The $(x_i + y_i - 1)^2$ term in the objective tries to push the variables towards satisfying $x_i + y_i = 1$. The $y_i \ge 0$ variable occurs only in this term and allows $x_i$ to take values $\le 1$. So, overall the $(x_i + y_i - 1)^2$ term acts as a soft $x_i \le 1$ constraint. The variable $z$ is to satisfy the simplex condition if $\sum_i x_i + \sum_i y_i < 2n$.

A $\delta$-approximate KKT point $(\bx, \by, z)$ of the above QP \eqref{eq:qp:1}--\eqref{eq:qp:3} satisfies the following conditions:
\begin{align}
    0 &\ge u - \delta, & \text{if } z = 0, \label{eq:kkt:z0} \\
    0 &\in [u - \delta, u + \delta], & \text{if } z > 0, \label{eq:kkt:z1} \\
    \frac{M}{\delta} (x_i + y_i - 1) &\ge u - \delta, & \text{if } y_i = 0, \label{eq:kkt:y0} \\
    \frac{M}{\delta} (x_i + y_i - 1) &\in [u - \delta, u + \delta], & \text{if } y_i > 0, \label{eq:kkt:y1} \\
    \sum_j a_{ij} x_j + b_i + \frac{M}{\delta} (x_i + y_i - 1) &\ge u - \delta, & \text{if } x_i = 0, \label{eq:kkt:x0} \\
    \sum_j a_{ij} x_j + b_i + \frac{M}{\delta} (x_i + y_i - 1) &\in [u - \delta, u + \delta], & \text{if } x_i > 0, \label{eq:kkt:x1}
\end{align}
for all $i \in [n]$ and where $u \in \bR$.

\begin{lemma}\label{lm:z}
$x_i + y_i < 2$ for all $i \in [n]$, which also implies $z > 0$ and $u \in [-\delta, \delta]$.
\end{lemma}
\begin{proof}
Irrespective of whether $z = 0$ or $z > 0$, from \eqref{eq:kkt:z0} and \eqref{eq:kkt:z1}, we have $u \le \delta$.
For contradiction, let $x_i + y_i \ge 2$ for some $i \in [n]$. Then either $x_i > 0$ or $y_i > 0$. If $y_i > 0$, then using \eqref{eq:kkt:y1}, we have
\[
    \frac{M}{\delta} (x_i + y_i - 1) \le u + \delta \implies \frac{M}{\delta} \le 2\delta \implies \delta^2 \ge \frac{M}{2},
\]
which is a contradiction because $\delta \le 1/4$ and $M \ge 1$ by definition. Similarly, if $x_i > 0$, then using \eqref{eq:kkt:x1}, we have
\begin{align*}
    \sum_j a_{ij} x_j + b_i + \frac{M}{\delta} (x_i + y_i - 1) &\le u + \delta \\
    \implies \sum_j (- |a_{ij}|) x_j - |b_i| + \frac{M}{\delta} &\le 2 \delta \quad\text{because $u \le \delta$, $x_i + y_i \ge 2$, and $x_j \geq 0,$}\\
    \implies -2nM + \frac{M}{\delta} &\le 2 \delta \qquad\qquad\qquad\quad\text{because $x_j \le 2n$ for all $j$,}\\
    \implies M \le 2\delta^2 + \delta 2nM &\le \frac{1}{8} + \frac{1}{2} \qquad\qquad\qquad\text{ because $\delta \le \frac{1}{4 + 4nM}$},
\end{align*}
which is a contradiction because $M \ge 1$. So, $x_i + y_i < 2$ for all $i \in [n]$, which implies that $\sum_i (x_i + y_i) < 2n \implies z > 0$ using \eqref{eq:qp:2}, which implies $u \in [-\delta, \delta]$ using \eqref{eq:kkt:z1}.
\end{proof}

From \cref{lm:z}, we know that $u \in [-\delta, \delta]$, using this to get rid of $u$ in \eqref{eq:kkt:y0}--\eqref{eq:kkt:x1}, we get the following conditions
\begin{align}
    \frac{M}{\delta} (x_i + y_i - 1) &\ge - 2\delta, & \text{if } y_i = 0, \label{eq:kkt:y0b} \\
    \frac{M}{\delta} (x_i + y_i - 1) &\in [- 2\delta, 2\delta], & \text{if } y_i > 0, \label{eq:kkt:y1b} \\
    \sum_j a_{ij} x_j + b_i + \frac{M}{\delta} (x_i + y_i - 1) &\ge -2\delta, & \text{if } x_i = 0, \label{eq:kkt:x0b} \\
    \sum_j a_{ij} x_j + b_i + \frac{M}{\delta} (x_i + y_i - 1) &\in [-2\delta, 2\delta], & \text{if } x_i > 0. \label{eq:kkt:x1b}
\end{align}

We now try to get rid of the dependency on $y_i$ in the KKT conditions \eqref{eq:kkt:y0b}--\eqref{eq:kkt:x1b}, albeit making these inequalities slightly weaker.
\begin{lemma}\label{lm:x}
For all $i \in [n]$, we have
\begin{align}
    \sum_j a_{ij} x_j + b_i &\ge -4\delta, & \text{if } x_i = 0, \label{eq:kkt:x0c} \\
    \sum_j a_{ij} x_j + b_i &\in [-4\delta, 4\delta], & \text{if } x_i \in (0,1), \label{eq:kkt:x1c} \\
    \sum_j a_{ij} x_j + b_i + \frac{M}{\delta} (x_i - 1) &\le 2\delta, & \text{if } x_i \ge 1. \label{eq:kkt:x2c}
\end{align}
\end{lemma}

\begin{proof}
We consider the following two cases:
\begin{enumerate}
    \item $x_i + y_i < 1$.
    As $x_i + y_i - 1 < 0$, we have $\frac{M}{\delta} (x_i + y_i - 1) < 0$. Plugging this in \eqref{eq:kkt:x0b} and \eqref{eq:kkt:x1b}, we have for all $x_i \ge 0$ (and $x_i < 1$ by the case assumption)
    \begin{equation*}\label{eq:kkt:x2a}
        \sum_j a_{ij} x_j + b_i + \frac{M}{\delta} (x_i + y_i - 1) \ge -2\delta \implies \sum_j a_{ij} x_j + b_i \ge -2\delta.
    \end{equation*}
    Further, from \eqref{eq:kkt:y0b} and \eqref{eq:kkt:y1b}, we know that $\frac{M}{\delta} (x_i + y_i - 1) \ge - 2\delta$ for all $y_i \ge 0$. Plugging this into \eqref{eq:kkt:x1b}, if $x_i > 0$, we have
    \begin{align*}\label{eq:kkt:x2b}
        \sum_j a_{ij} x_j + b_i + \frac{M}{\delta} (x_i + y_i - 1) &\le 2\delta \\
        \implies \sum_j a_{ij} x_j + b_i &\le 2\delta - \frac{M}{\delta} (x_i + y_i - 1) \le 4\delta.
    \end{align*}

    \item $x_i + y_i \ge 1$. 
    If $y_i = 0$, we know that $x_i \ge 1$, and \eqref{eq:kkt:x1b} reduces to
    \begin{equation*}
        \sum_j a_{ij} x_j + b_i + \frac{M}{\delta} (x_i - 1) \in [-2\delta, 2\delta].
    \end{equation*}
    If $y_i > 0$, from \eqref{eq:kkt:y1b} we have $\frac{M}{\delta} (x_i + y_i - 1) \in [-2\delta, 2\delta]$. Plugging this into \eqref{eq:kkt:x0b} and \eqref{eq:kkt:x1b}, we get
    \begin{align}
        \sum_j a_{ij} x_j + b_i &\ge -4\delta, & \text{if } x_i = 0, \label{eq:kkt:x2d} \\
    \sum_j a_{ij} x_j + b_i &\in [-4\delta, 4\delta], & \text{if } x_i > 0. \label{eq:kkt:x2e}
    \end{align}
    Further, as $y_i > 0$, from \eqref{eq:kkt:x1b} we also trivially have for $x_i \ge 1$
    \begin{align*}
        \sum_j a_{ij} x_j + b_i + \frac{M}{\delta} (x_i - 1) &= \sum_j a_{ij} x_j + b_i + \frac{M}{\delta} (x_i + y_i - 1) - \frac{M}{\delta}y_i  \\
        &\le 2\delta - \frac{M}{\delta}y_i \le 2\delta.
    \end{align*}
\end{enumerate}
\end{proof}

We propose $\bx' = (x_1', \ldots, x_n')$, where $x_i' = \min(1, x_i)$, as an $\varepsilon$-approximate KKT point of the original box-constrained QP.

If $x_i \le 1$ for all $i$, then from \cref{lm:x}, we are already at a $4\delta \le \varepsilon$ approximate KKT point. 

Let us now consider the case where $x_i > 1$ for some $i$.
As $x_i > 1$, from \cref{lm:x}, we have
\begin{align*}
    2\delta &\ge \sum_j a_{ij} x_j + b_i + \frac{M}{\delta} (x_i - 1) \ge -\sum_j |a_{ij}| x_j - |b_i| + \frac{M}{\delta} (x_i - 1) \\
    &\ge -\sum_j 2|a_{ij}|  - |b_i| + \frac{M}{\delta} (x_i - 1) \qquad \text{ as $x_j < 2$ from \cref{lm:z} }\\
    &\ge -2M + \frac{M}{\delta} (x_i - 1) \\
    \implies x_i &\le 1 + 2 \delta \left(1 + \frac{\delta}{M}\right) \le 1 + 4\delta.
\end{align*}
This implies that for all $i$, $x_i' \in [x_i - 4\delta, x_i]$. So, for all $i$, $\sum_j a_{ij} (x_j - x_j') \le \sum_j |a_{ij}| |x_j - x_j'| \le 4 \delta M $. Similarly, for all $i$, $\sum_j a_{ij} (x_j - x_j') \ge -4 \delta M$. Putting this together with \eqref{eq:kkt:x0c}--\eqref{eq:kkt:x2c} of \cref{lm:x}, we get
\begin{align*}
    \sum_j a_{ij} x_j' + b_i &\ge \sum_j a_{ij} x_j - 4\delta M  + b_i \ge -4\delta(1+M),\qquad \text{if } x_i' = x_i = 0, \\
    \sum_j a_{ij} x_j' + b_i &\in \sum_j a_{ij} x_j + [-4\delta M, 4\delta M] + b_i \in [-4\delta(1+M), 4\delta(1+M)], \\
    &\qquad\qquad\qquad\qquad\qquad\qquad\qquad\qquad\qquad\quad\ \ \text{if } x_i' = x_i \in (0,1), \\
    \sum_j a_{ij} x_j' + b_i &\le \sum_j a_{ij} x_j + 4\delta M + b_i + \frac{M}{\delta} (x_i - 1) \le 4\delta(1+M), \\
    &\qquad\qquad\qquad\qquad\qquad\qquad\qquad\qquad\qquad\quad\ \ \text{if } x_i' = 1 \Longleftrightarrow x_i \ge 1, 
\end{align*}
where on the last line we also used the fact that $\frac{M}{\delta} (x_i' - 1) = 0$ when $x_i' = 1$ and $\frac{M}{\delta} (x_i - 1) \geq 0$ when $x_i \geq 1$.
As $4\delta(1+M) \le \varepsilon$, $\bx'$ satisfies the required box-constraint KKT conditions, and we are done.
\end{proof}

%% file: main.bbl
\begin{thebibliography}{10}
\providecommand{\url}[1]{\texttt{#1}}
\providecommand{\urlprefix}{URL }
\providecommand{\doi}[1]{https://doi.org/#1}

\bibitem{AdsulGMS11-rank-one}
Adsul, B., Garg, J., Mehta, R., Sohoni, M.: Rank-1 bimatrix games: a homeomorphism and a polynomial time algorithm. In: Proceedings of the 43rd ACM Symposium on Theory of Computing (STOC). pp. 195--204 (2011). \doi{10.1145/1993636.1993664}

\bibitem{BabichenkoR21-congestion}
Babichenko, Y., Rubinstein, A.: Settling the complexity of {N}ash equilibrium in congestion games. In: Proceedings of the 53rd ACM Symposium on Theory of Computing (STOC). pp. 1426--1437 (2021). \doi{10.1145/3406325.3451039}

\bibitem{BitanskyPR15-Nash-crypto}
Bitansky, N., Paneth, O., Rosen, A.: On the cryptographic hardness of finding a {Nash} equilibrium. In: Proceedings of the 56th Symposium on Foundations of Computer Science (FOCS). pp. 1480--1498 (2015). \doi{10.1109/focs.2015.94}

\bibitem{ChenDT09-Nash}
Chen, X., Deng, X., Teng, S.H.: Settling the complexity of computing two-player {N}ash equilibria. Journal of the ACM  \textbf{56}(3),  14:1--14:57 (2009). \doi{10.1145/1516512.1516516}

\bibitem{ChoudhuriHKPRR19-Fiat-Shamir}
Choudhuri, A.R., Hub{\'{a}}{\v{c}}ek, P., Kamath, C., Pietrzak, K., Rosen, A., Rothblum, G.N.: Finding a {Nash} equilibrium is no easier than breaking {Fiat}-{Shamir}. In: Proceedings of the 51st ACM Symposium on Theory of Computing (STOC). pp. 1103--1114 (2019). \doi{10.1145/3313276.3316400}

\bibitem{CodenottiS05-bimatrix-zero-one-imitation}
Codenotti, B., Štefankovič, D.: On the computational complexity of {N}ash equilibria for bimatrix games. Information Processing Letters  \textbf{94}(3),  145--150 (2005). \doi{10.1016/j.ipl.2005.01.010}

\bibitem{ConitzerS08-Nash}
Conitzer, V., Sandholm, T.: New complexity results about {N}ash equilibria. Games and Economic Behavior  \textbf{63}(2),  621--641 (2008). \doi{10.1016/j.geb.2008.02.015}

\bibitem{DaskalakisGP09-Nash}
Daskalakis, C., Goldberg, P.W., Papadimitriou, C.H.: The complexity of computing a {N}ash equilibrium. SIAM Journal on Computing  \textbf{39}(1),  195--259 (2009). \doi{10.1137/070699652}

\bibitem{DaskalakisP11-CLS}
Daskalakis, C., Papadimitriou, C.: Continuous local search. In: Proceedings of the 22nd ACM-SIAM Symposium on Discrete Algorithms (SODA). pp. 790--804 (2011). \doi{10.1137/1.9781611973082.62}

\bibitem{Emmons22-symmetric-common-payoff}
Emmons, S., Oesterheld, C., Critch, A., Conitzer, V., Russell, S.: For learning in symmetric teams, local optima are global {N}ash equilibria. In: Proceedings of the 39th International Conference on Machine Learning (ICML). pp. 5924--5943 (2022), \url{https://proceedings.mlr.press/v162/emmons22a.html}

\bibitem{EtessamiY10-FIXP}
Etessami, K., Yannakakis, M.: On the complexity of {N}ash equilibria and other fixed points. SIAM Journal on Computing  \textbf{39}(6),  2531--2597 (2010). \doi{10.1137/080720826}

\bibitem{FabrikantPT2004pure}
Fabrikant, A., Papadimitriou, C., Talwar, K.: {The complexity of pure Nash equilibria}. In: Proceedings of the 36th Annual ACM Symposium on Theory of Computing (STOC). pp. 604--612. ACM (2004). \doi{10.1145/1007352.1007445}

\bibitem{FearnleyGHS22-gradient}
Fearnley, J., Goldberg, P., Hollender, A., Savani, R.: The complexity of gradient descent: {CLS} = {PPAD} $\cap$ {PLS}. Journal of the ACM  \textbf{70}(1),  7:1--7:74 (2022). \doi{10.1145/3568163}

\bibitem{fearnley2024complexity}
Fearnley, J., Goldberg, P.W., Hollender, A., Savani, R.: The complexity of computing {KKT} solutions of quadratic programs. In: Proceedings of the 56th ACM Symposium on Theory of Computing (STOC). pp. 892--903 (2024). \doi{10.1145/3618260.3649647}

\bibitem{GilboaZ1989-Nash-decision}
Gilboa, I., Zemel, E.: Nash and correlated equilibria: Some complexity considerations. Games and Economic Behavior  \textbf{1}(1),  80--93 (1989). \doi{10.1016/0899-8256(89)90006-7}

\bibitem{JawaleKKZ21-PPAD-LWE}
Jawale, R., Kalai, Y.T., Khurana, D., Zhang, R.: {SNARG}s for bounded depth computations and {PPAD} hardness from sub-exponential {LWE}. In: Proceedings of the 53rd ACM Symposium on Theory of Computing (STOC). pp. 708--721 (2021). \doi{10.1145/3406325.3451055}

\bibitem{JPY1988-PLS}
Johnson, D.S., Papadimitriou, C.H., Yannakakis, M.: How easy is local search? Journal of Computer and System Sciences  \textbf{37}(1),  79--100 (1988). \doi{10.1016/0022-0000(88)90046-3}

\bibitem{KontogiannisS11-symmetric-bimatrix}
Kontogiannis, S., Spirakis, P.: Approximability of symmetric bimatrix games and related experiments. In: Proceedings of the 10th International Symposium on Experimental Algorithms. pp. 1--20 (2011). \doi{10.1007/978-3-642-20662-7_1}

\bibitem{Krentel1989-TSP}
Krentel, M.W.: Structure in locally optimal solutions. In: Proceedings of the 30th Annual Symposium on Foundations of Computer Science (FOCS). pp. 216--221 (1989). \doi{10.1109/SFCS.1989.63481}

\bibitem{LiuS18-sparse-win-lose}
Liu, Z., Sheng, Y.: On the approximation of {Nash} equilibria in sparse win-lose games. Proceedings of the 32nd AAAI Conference on Artificial Intelligence (AAAI) pp. 1154--1160 (2018). \doi{10.1609/aaai.v32i1.11439}

\bibitem{McLennanT10-complexity-imitation}
McLennan, A., Tourky, R.: Simple complexity from imitation games. Games and Economic Behavior  \textbf{68}(2),  683--688 (2010). \doi{10.1016/j.geb.2009.10.003}

\bibitem{MegiddoP91-TFNP}
Megiddo, N., Papadimitriou, C.H.: On total functions, existence theorems and computational complexity. Theoretical Computer Science  \textbf{81}(2),  317--324 (1991). \doi{10.1016/0304-3975(91)90200-L}

\bibitem{Mehta18-constant-rank}
Mehta, R.: Constant rank two-player games are {PPAD}-hard. SIAM Journal on Computing  \textbf{47}(5),  1858--1887 (2018). \doi{10.1137/15m1032338}

\bibitem{Morioka01-Mthesis-PLS}
Morioka, T.: Classification of search problems and their definability in bounded arithmetic. Master's thesis, University of Toronto (2001), \url{https://www.collectionscanada.ca/obj/s4/f2/dsk3/ftp04/MQ58775.pdf}

\bibitem{MurhekarM20-imitation}
Murhekar, A., Mehta, R.: Approximate {N}ash equilibria of imitation games: Algorithms and complexity. In: Proceedings of the 19th International Conference on Autonomous Agents and Multiagent Systems (AAMAS). pp. 887--894 (2020). \doi{10.5555/3398761.3398865}

\bibitem{Nash51}
Nash, J.: Non-cooperative games. The Annals of Mathematics  \textbf{54}(2),  286--295 (1951). \doi{10.2307/1969529}

\bibitem{Nash50}
Nash, J.F.: Equilibrium points in $n$-person games. Proceedings of the National Academy of Sciences  \textbf{36}(1),  48--49 (1950). \doi{10.1073/pnas.36.1.48}

\bibitem{Neumann1928-zero-sum}
von Neumann, J.: {Zur Theorie der Gesellschaftsspiele}. Mathematische Annalen  \textbf{100}(1),  295--320 (1928). \doi{10.1007/bf01448847}

\bibitem{Papadimitriou94-TFNP-subclasses}
Papadimitriou, C.H.: On the complexity of the parity argument and other inefficient proofs of existence. Journal of Computer and System Sciences  \textbf{48}(3),  498--532 (1994). \doi{10.1016/S0022-0000(05)80063-7}

\bibitem{Papadimitriou07-AGT-complexity-equilibria}
Papadimitriou, C.H.: The complexity of finding {N}ash equilibria. In: Nisan, N., Roughgarden, T., Tardos, {\'E}., Vazirani, V.V. (eds.) Algorithmic Game Theory, pp. 29--52. Cambridge University Press (2007). \doi{10.1017/cbo9780511800481.004}

\bibitem{Schaeffer1991local}
Sch{\"a}ffer, A.A., Yannakakis, M.: Simple local search problems that are hard to solve. SIAM Journal on Computing  \textbf{20}(1),  56--87 (1991). \doi{10.1137/0220004}

\bibitem{TewoldeZOSC25-game-symmetries}
Tewolde, E., Zhang, B.H., Oesterheld, C., Sandholm, T., Conitzer, V.: Computing game symmetries and equilibria that respect them. Proceedings of the 39th AAAI Conference on Artificial Intelligence (AAAI) pp. 14148--14157 (2025). \doi{10.1609/aaai.v39i13.33549}

\end{thebibliography}
